\documentclass{IEEEtran}
\usepackage{amsmath}
\usepackage{amsthm} 
\usepackage{amssymb} 
\usepackage{breqn}
\usepackage{setspace}
\usepackage{graphicx}
\usepackage[T1]{fontenc}
\usepackage{supertabular}
\usepackage{longtable}
\usepackage{mdframed}
\usepackage{fancyhdr}
\newtheorem{theorem}{Theorem}
\newtheorem{lemma}{Lemma}

\newcommand{\mathsym}[1]{{}}
\newcommand{\unicode}[1]{{}}

\begin{document}
%
\title{Unique Option Pricing Measure With Neither Dynamic Hedging nor Complete Markets}

\author{
    \IEEEauthorblockN{Nassim Nicholas Taleb\IEEEauthorrefmark{1}\IEEEauthorrefmark{2}}
    \IEEEauthorblockA{\IEEEauthorrefmark{1}School of Engineering, NYU, \& Former Option Trader\\}
}

\maketitle 

\begin{abstract} Proof that under simple assumptions, such as constraints of Put-Call Parity, the probability measure for the valuation of a European option has the mean derived from the forward price which can, but does not have to be the risk-neutral one, under any general probability distribution, bypassing the Black-Scholes-Merton dynamic hedging argument, and without the requirement of complete markets and other strong assumptions. We confirm that the heuristics used by traders for centuries are both more robust, more consistent, and more rigorous than held in the economics literature. We also show that options can be priced using infinite variance (finite mean) distributions. 

\end{abstract}
\thispagestyle{fancy}
\markboth{\textbf{Extreme Risk Initiative ---NYU School of Engineering Working Paper Series}}
\maketitle
\section{Background}
Option valuations methodologies have been used by traders for centuries, in an effective way (Haug and Taleb, 2010). In addition, valuations by expectation of terminal payoff forces the mean of the probability distribution used for option prices be be that of the forward, thanks to Put-Call Parity and, should the forward be risk-neutrally priced, so will the option be. The Black Scholes argument (Black and Scholes, 1973, Merton, 1973) is held to allow  risk-neutral option pricing thanks to dynamic hedging, as the option becomes redundant (since its payoff can be built as a linear combination of cash and the underlying asset dynamically revised through time). This is a puzzle, since: 1) Dynamic Hedging is not operationally feasible in financial markets owing to the dominance of portfolio changes resulting from jumps, 2) The dynamic hedging argument doesn't stand mathematically under fat tails; it requires a very specific "Black Scholes world" with many impossible assumptions, one of which requires finite quadratic variations, 3) Traders use the same Black-Scholes "risk neutral argument" for the valuation of options on assets that do not allow dynamic replication, 4) Traders trade options consistently in domain where the risk-neutral arguments do not apply 5) There are fundamental informational limits preventing the convergence of the stochastic integral.\footnote{Further, in a case of scientific puzzle, the exact formula called "Black-Scholes-Merton" was written down (and used) by Edward Thorp in a heuristic derivation by expectation that did not require dynamic hedging, see Thorpe(1973).}

There have been a couple of predecessors to the present thesis that Put-Call parity is sufficient constraint to enforce some structure at the level of the mean of the underlying distribution, such as Derman and Taleb (2005), Haug and Taleb (2010). These approaches were heuristic, robust though deemed hand-waving (Ruffino and Treussard, 2006). In addition they showed that operators need to use the risk-neutral mean. What this paper does is  \begin{itemize}
  \item It goes beyond the "handwaving" with formal proofs.
   \item It uses a completely distribution-free, expectation-based approach and proves the risk-neutral argument without dynamic hedging, and without any distributional assumption. 
  \item Beyond risk-neutrality, it establishes the case of a unique pricing distribution for option prices in the absence of such argument. The forward (or future) price can embed expectations and deviate from the arbitrage price (owing to, say, regulatory or other limitations) yet the options can still be priced at a distibution corresponding to the mean of such a forward.
  \item It shows how one can \textit{practically} have an option market without "completeness" and without having the theorems of financial economics hold.
 
  \end{itemize}

These are done with solely two constraints: "horizontal", i.e. put-call parity, and "vertical", i.e. the different valuations across strike prices deliver a probability measure which is shown to be unique.  The only economic assumption made here is that the forward exits,  is tradable --- in the absence of such unique forward price it is futile to discuss standard option pricing. We also require the probability measures to correspond to distributions with finite first moment.

Preceding works in that direction are as follows. Breeden and Litzenberger (1978) and Dupire(1994), show how option spreads deliver a unique probability measure; there are papers establishing broader set of arbitrage relations between options such as Carr and Madan (2001)\footnote{See also Green and Jarrow (1987) and  Nachman(1988). We have known about the possibility of risk neutral pricing without dynamic hedging since Harrison and Kreps (1979) but the theory necessitates extremely strong --and severely unrealistic --assumptions, such as strictly complete markets and a multiperiod pricing kernel}.

However 1) none of these papers made the bridge between calls and puts via the forward, thus translating the relationships from arbitrage relations between options delivering a probability distribution into the necessity of lining up to the mean of the distribution of the forward, hence the risk-neutral one (in case the forward is arbitraged.) 2) Nor did any paper show that in the absence of second moment (say, infinite variance), we can price options very easily. Our methodology and proofs make no use of the variance. 3) Our method is vastly simpler, more direct, and robust to changes in assumptions.

We make no assumption of general market completeness. Options are not redundant securities and remain so. Table 1 summarizes the gist of the paper.\footnote{The famed Hakkanson paradox is as follows: if markets are complete and options are redudant, why would someone need them? If markets are incomplete, we may need options but how can we price them? This discussion may have provided a solution to the paradox: markets are incomplete \textit{and} we can price options.} \footnote{Option prices are not unique in the absolute sense: the premium over intrinsic can take an entire spectrum of values; it is just that the put-call parity constraints forces the measures used for puts and the calls to be the same and to have the same expectation as the forward. As far as securities go, options are securities on their own; they just have a strong link to the forward.} 
\begin{table}
\caption{Main practical differences between the dynamic hedging argument and the static Put-Call parity with speading across strikes.}
    \begin{tabular}{ | p{2. cm}| p{2.6cm}|p{2.6cm}|}
    \hline
    & & \\
    \textbf{}  & \textbf{Black-Scholes Merton} & \textbf{Put-Call Parity with Spreading}\\ 
\hline 
& & \\
\textbf{Type } & Continuous rebalancing.  & Interpolative static hedge.\\
\hline
& &  \\

\textbf{Market Assumptions} & 1) Continuous Markets, no gaps, no jumps.  & 1) Gaps and jumps acceptable. Continuous Strikes, or acceptable number of strikes.\\
& &  \\

 & 2) Ability to borrow and lend underlying asset for all dates. & 2) Ability to borrow and lend underlying asset for single forward date.\\ 
 
& &  \\
& 3) No transaction costs in trading asset. & 3) Low transaction costs in trading options. \\
& & \\

\hline
& & \\
\textbf{Probability Distribution} & Requires all moments to be finite. Excludes the class of slowly varying distributions & Requires finite $1^{st}$ moment (infinite variance is acceptable).\\ 
\hline
& &  \\

\textbf{Market Completeness} & Achieved through dynamic completeness & Not required (in the traditional sense)\\ 
\hline

& & \\
\textbf{Realism of Assumptions} & Low  & High\\ 
\hline

& & \\
\textbf{Convergence} & In probability (uncertain; one large jump changes expectation)  & Pointwise\\ 
\hline

\hline
\textbf{Fitness to Reality} & Only used after "fudging" standard deviations per strike.  & Portmanteau, using specific distribution adapted to reality\\ 
& &  \\
\hline

    \end{tabular}

    \end{table}

\section{Proof}
Define $C(S_{t_0},K,t)$ and $P(S_{t_0},K,t)$ as European-style  call and put with strike price K, respectively, with expiration $t$, and $S_0$ as an underlying security at times $t_0$, $t \geq t_0$, and $S_t$ the possible value of the underlying security at time t. 

\subsection{Case 1: Forward as risk-neutral measure}

Define $r = \frac{1}{t-t_0}\int_{t_0}^t r_s \mathrm{d}s$, the return of a risk-free money market fund and $\delta =\frac{1}{t-t_0}\int_{t_0}^t \delta_s \mathrm{d}s$ the payout of the asset (continuous dividend for a stock, foreign interest for a currency).

We have the arbitrage forward price $F_t^Q$:
\begin{equation}
 F_t^Q = S_0\frac{(1+r)^{(t-t_0)}}{(1+\delta)^{(t-t_0)}} \thickapprox S_0 \, e^{(r-\delta) (t-t_0)}
\end{equation}
by arbitrage, see Keynes 1924. We thus call $F_t^Q$ the future (or forward) price obtained by arbitrage, at the risk-neutral rate. Let $F_t^P$ be the future requiring a risk-associated "expected return" $m$, with expected forward price:
\begin{equation}
 F_t^P = S_0 (1+m)^{(t-t_0)}\thickapprox S_0 \, e^{m \, (t-t_0)}.
\end{equation}

\noindent
\textbf{Remark:  }\textit{ By arbitrage, all tradable values of the forward price given $S_{t_0}$ need to be equal to $F_t^Q$.}

"Tradable" here does not mean "traded", only subject to arbitrage replication by "cash and carry", that is, borrowing cash and owning the secutity yielding $d$ if the embedded forward return diverges from $r$.

\subsection{Derivations}

In the following we take $F$ as having dynamics on its own --irrelevant to whether we are in case 1 or 2 --hence a unique probability measure $Q$.

Define $\Omega=[0,\infty)=A_K \cup A_K^c $
where 
$A_K=[0,K]$ and $A_K^c=(K, \infty)$.

Consider a class of standard (simplified) probability spaces $(\Omega,\mu_i)$ indexed by $i$, where 
$\mu_i$ is a probability measure, i.e., satisfying $\int_\Omega \mathrm{d} \mu_i=1$.

\begin{theorem}
For a given maturity T, there is a unique measure $\mu_Q$ that prices European puts and calls by expectation of terminal payoff. 
\end{theorem}

This measure can be risk-neutral in the sense that it prices the forward $F_t^Q$, but does not have to be and imparts rate of return to the stock embedded in the forward.

\begin{lemma}
For a given maturity T, there exist two measures $\mu_1$ and $\mu_2$ for European calls and puts of the same maturity and same underlying security associated with the valuation by expectation of terminal payoff, which are unique such that, for any call and put of strike K, we have:
\begin{equation}
 C= \int_ \Omega f_C \, \mathrm{d}\mu_1 \,  ,\label{callequation}
\end{equation}
and 
\begin{equation}
  P= \int_ \Omega f_P \, \mathrm{d}\mu_2 \, ,
\end{equation}
respectively, and where $f_C$ and $f_P$ are $(S_t-K)^+$ and $(K-S_t)^+$ respectively.

\end{lemma}

\begin{proof}
For clarity, set $r$ and $\delta$ to $0$ without a loss of generality.
By Put-Call Parity Arbitrage, a positive holding of a call ("long") and negative one of a put ("short") replicates a tradable forward; because of P/L variations, using positive sign for long and negative sign for short:

\begin{equation}
C(S_{t_0},K,t)-P(S_{t_0},K,t)+K=F_t^P \label{putcallparity}
\end{equation}
necessarily since $F_t^P$ is tradable. 

Put-Call Parity holds for all strikes, so:
\begin{equation}
C(S_{t_0},K +\Delta K,t)-P(S_{t_0},K+\Delta K,t)+K+\Delta K=F_t^P\label{diffputcallparity}
\end{equation}
for all $K \in \Omega $

Now a Call spread in quantities $\frac{1}{\Delta K}$, expressed as \[C(S_{t_0},K,t)-C(S_{t_0},K+\Delta K,t),\]delivers \$1 if $S_t > K+\Delta K$ (that is, corresponds to the indicator function $\mathbf{1}_{S > K+\Delta K}$), 0 if $S_t\leq K$ (or $\mathbf{1}_{S > K}$), and the quantity times $S_t-K$ if $K < S_t \leq K+\Delta K$, that is, between 0 and \$1 (see Breeden and Litzenberger, 1978).  Likewise, consider the converse argument for a put, with $\Delta K <S_t$.

At the limit, for $\Delta K \to  0$
\begin{equation}
 \frac{\partial{C(S_{t_0},K,t)}}{\partial{K}} =- P(S_t > K) 
 =- \int_{A_K^c} \mathrm{d}\mu_1 .
\end{equation}
By the same argument:
\begin{equation}
 \frac{\partial{P(S_{t_0},K,t)}}{\partial{K}} = \int_{A_K} \mathrm{d}\mu_2  =  1-\int_{A_K^c} \mathrm{d}\mu_2 .
\end{equation}

As semi-closed intervals generate the whole Borel $\sigma$-algebra on $\Omega$, this shows that $\mu_1$and $\mu_2$ are unique.

\end{proof}
\begin{lemma}
The probability measures of puts and calls are the same, namely  for each Borel set $A$ in $\Omega$, $\mu_1(A)$ = $\mu_2(A)$. 
\end{lemma}

\begin{proof}
Combining Equations \ref{putcallparity} and \ref{diffputcallparity}, dividing by $\frac{1}{\Delta K}$ and taking $\Delta K \to  0$:

\begin{equation}
-\frac{\partial{C(S_{t_0},K,t)}}{\partial{K}}+\frac{\partial{P(S_{t_0},K,t)}}{\partial{K}}=1
\end{equation}
for all values of $K$, so

\begin{equation}
\int_{A_K^c} \mathrm{d}\mu_1=\int_{A_K^c} \mathrm{d}\mu_2 ,
\end{equation}
hence $\mu_1(A_K)=\mu_2(A_K)$ for all $K \in [0,\infty)$. This equality being true for any semi-closed interval, it extends to any Borel set.

 \begin{equation*}
\qedhere
 \end{equation*}
\end{proof}

\begin{lemma}
Puts and calls are required, by static arbitrage, to be evaluated at same as risk-neutral measure $\mu_Q$ as the tradable forward.
\end{lemma}

\begin{proof}
\begin{equation}
F_t^P=\int_\Omega F_t \, \mathrm{d}\mu_Q;
\end{equation}

from Equation \ref{putcallparity}

\begin{equation}
\int_\Omega f_C(K) \,\mathrm{d}\mu_1- \int_\Omega f_P(K)\,\mathrm{d}\mu_1 = \int_\Omega F_t \, \mathrm{d}\mu_Q -K
\end{equation}

Taking derivatives on both sides, and since $f_C-f_P=S_0+K$, we get the Radon-Nikodym derivative:
\begin{equation}
\frac{\mathrm{d}\mu_Q}{\mathrm{d}\mu_1}=1
\end{equation}
for all values of K.
 
\end{proof}

\[\qedhere\]

\section{Case where the Forward is not risk neutral}
Consider the case where  $F_t$ is observable, tradable, and use it solely as an underlying security with dynamics on its own. In such a case we can completely ignore the dynamics of the nominal underlying $S$, or use a non-risk neutral "implied" rate linking cash to forward, $m^*= \frac{\log \left(\frac{F}{S_0}\right)}{t-t_0}$. the rate $m$ can embed risk premium, difficulties in financing, structural or regulatory impediments to borrowing, with no effect on the final result.

In that situation, it can be shown that the exact same results as before apply, by remplacing the measure $\mu_Q$ by another measure $\mu_{Q^*}$. Option prices remain unique \footnote{We assumed 0 discount rate for the proofs; in case of nonzero rate, premia are discounted at the rate of the arbitrage operator}.

\section{comment}
We have replaced the complexity and intractability of dynamic hedging with a simple, more benign interpolation problem, and explained the performance of pre-Black-Scholes option operators using simple heuristics and rules, bypassing the structure of the theorems of financial economics.

Options can remain non-redundant and markets incomplete: we are just arguing here for a form of arbitrage pricing (which includes risk-neutral pricing at the level of the expectation of the probability measure), nothing more. But this is sufficient for us to use any probability distribution with finite first moment, which includes the Lognormal, which recovers Black Scholes.

A final comparison. In dynamic heding, missing a single hedge, or encountering a single gap (a tail event) can be disastrous ---as we mentioned, it requires a series of assumptions beyond the mathematical, in addition to severe and highly unrealistic constraints on the mathematical. Under the class of fat tailed distributions, increasing the frequency of the hedges does not guarantee reduction of risk. Further, the standard dynamic hedging argument requires the exact specification of the \textit{risk-neutral} stochastic process between $t_0$ and $t$, something econometrically unwieldy, and which is generally reverse engineered from the price of options, as an arbitrage-oriented interpolation tool rather than as a representation of the process. 


Here, in our Put-Call Parity based methodology, our ability to track the risk neutral distribution is guaranteed by adding strike prices, and since probabilities add up to 1, the degrees of freedom that the recovered measure $\mu_Q$ has in the gap area between a strike price \(K\) and the next strike up, \(K +\Delta K\), are severely reduced, since the measure in the interval is constrained by the difference $\int_{A_K}^c \mathrm{d}\mu - \int_{A_{K+\Delta K}}^c \mathrm{d}\mu $. In other words, no single gap between strikes can significantly affect the probability measure, even less the first moment, unlike with dynamic hedging. In fact it is no different from standard kernel smoothing methods for statistical samples, but applied to the distribution across strikes.\footnote{For methods of interpolation of implied probability distribution between strikes, see Avellaneda et al.(1997).}

The assumption about the presence of strike prices constitutes a natural condition: conditional on having a \textit{practical} discussion about options, options strikes need to exist. Further, as it is the experience of the author, market-makers can add over-the-counter strikes at will, should they need to do so.

%
%
%
%
%

\section*{Acknowledgment}
Peter Carr, Marco Avellaneda, H\'elyette Geman, Raphael Douady, Gur Huberman, Espen Haug, and Hossein Kazemi.
\section*{References}

\noindent
Avellaneda, M., Friedman, C., Holmes, R., \& Samperi, D. (1997). Calibrating volatility surfaces via relative-entropy minimization. Applied Mathematical Finance, 4(1), 37-64.

\smallskip
\noindent
Black, F., Scholes, M. (1973). The pricing of options and corporate liabilities. Journal of Political Economy 81, 637-654.

\smallskip
\noindent
Breeden, D. T., \& Litzenberger, R. H. (1978). Prices of state-contingent claims implicit in option prices. Journal of business, 621-651.

\smallskip
\noindent
Carr, P. and Madan, D. (2001). Optimal positioning in derivative securities, Quantitative Finance, pp. 19-37.

\smallskip
\noindent
Derman, E. and Taleb, N. (2005). The illusions of dynamic replication. Quantitative Finance, 5(4):323-326.

\smallskip
\noindent
Dupire, Bruno, 1994, Pricing with a smile, Risk 7, 18-20.

\smallskip
\noindent
Green, R. C., \& Jarrow, R. A. (1987). Spanning and completeness in markets with contingent claims. Journal of Economic Theory, 41(1), 202-210.

\smallskip
\noindent
Harrison, J. M., \& Kreps, D. M. (1979). Martingales and arbitrage in multiperiod securities markets. Journal of Economic theory, 20(3), 381-408.

\smallskip
\noindent
Haug, E. G. and Taleb, N. N. (2010) Option Traders use Heuristics, Never the Formula known as Black-Scholes-Merton Formula, Journal of Economic Behavior and Organizations, pp. 97–106. 

\smallskip
\noindent
Keynes, J.M., 1924. A Tract on Monetary Reform. Reprinted in 2000. Prometheus Books, Amherst New York.

\smallskip
\noindent
Merton, R.C., 1973. Theory of rational option pricing. Bell Journal of Economics and Management Science 4, 141-183.

\smallskip
\noindent
Nachman, D. C. (1988). Spanning and completeness with options. Review of Financial Studies, 1(3), 311-328.

\smallskip
\noindent
Ruffino, D., \& Treussard, J. (2006). Derman and Taleb's "The illusions of dynamic replication": a comment. Quantitative Finance, 6(5), 365-367.

\smallskip
\noindent
Thorp, E.O., 1973. A corrected derivation of the Black-Scholes option model. In:  CRSP proceedings, 1976.

\end{document}